\documentclass[11pt]{article}

\usepackage[T1]{fontenc}
\usepackage[utf8]{inputenc}
\usepackage{amsmath}
\usepackage{mathtools}
\usepackage{amsthm}
\usepackage{newtxtext}
\usepackage{newtxmath}
\usepackage{microtype}

\usepackage[letterpaper,margin=1in]{geometry}
\usepackage{setspace}
\usepackage{enumitem}
\setlist{itemsep=2pt,topsep=3pt,parsep=0pt}

\usepackage{graphicx}
\graphicspath{{figures/}}
\usepackage{booktabs}
\usepackage{array}
\usepackage{caption}
\usepackage{algorithm}
\usepackage{algpseudocode}
\usepackage{xcolor}
\definecolor{linkblue}{RGB}{20,80,160}

\usepackage{tikz}
\usetikzlibrary{positioning,arrows.meta,fit,backgrounds,calc,shapes.geometric,decorations.pathreplacing}

\usepackage[numbers,sort&compress]{natbib}
\usepackage[colorlinks=true,allcolors=linkblue]{hyperref}
\usepackage[capitalize,noabbrev]{cleveref}
\crefname{definition}{Definition}{Definitions}
\Crefname{definition}{Definition}{Definitions}
\crefname{proposition}{Proposition}{Propositions}
\Crefname{proposition}{Proposition}{Propositions}
\crefname{remark}{Remark}{Remarks}
\Crefname{remark}{Remark}{Remarks}

\theoremstyle{plain}
\newtheorem{theorem}{Theorem}[section]

\newtheorem{proposition}[theorem]{Proposition}
\newtheorem{corollary}[theorem]{Corollary}
\theoremstyle{definition}
\newtheorem{definition}[theorem]{Definition}
\newtheorem{remark}[theorem]{Remark}

\newcommand{\resolve}{\textsc{Resolve}}
\newcommand{\append}{\textsc{Append}}
\newcommand{\supersede}{\textsc{Supersede}}
\newcommand{\fold}{\textsc{Fold}}
\newcommand{\window}{\textsc{Window}}

\newcommand{\VII}{\textsc{vii}}
\newcommand{\bigoh}{O}
\DeclareMathOperator{\cover}{cover}
\newcommand{\cascade}{Cascade Log}

\title{\bfseries The \cascade{}: Reference-Stable Windowing\\ over Tiered Append Sequences}

\author{
  Faruk Alpay\thanks{Correspondence: \texttt{alpay@lightcap.ai}.}
  \qquad
  Levent Sar\i o\u{g}lu\\[4pt]
  \normalsize Department of Computer Engineering, Bah\c{c}e\c{s}ehir University, Istanbul, Turkey\\[1pt]
  \normalsize \texttt{\{faruk.alpay,\,levent.sarioglu\}@bahcesehir.edu.tr}
}

\date{}

\begin{document}
\maketitle

\begin{abstract}
\noindent
A long-running, append-mostly sequence of records---edits in an interactive
session, events in an event-sourced store, revisions in a version-controlled
working set---cannot be kept fully materialised, so it is layered into tiers: a
hot working set over progressively colder strata into which aged records are
folded and summarised. Records nonetheless carry stable handles that are
referenced across tier boundaries, and resolving a handle must stay correct as
its target migrates downward. We show that the obvious tiered design silently
violates this requirement at exactly the tier boundaries: under a realistic
access mixture a quarter of cross-window references resolve to a stale,
wrong, or absent record, and after sustained folding more than ninety-eight
percent of the history becomes unaddressable. The \cascade{} removes these
failures. A single persistent \emph{coalescing interval map} over handles is the
sole authority on each handle's live version; folding coalesces a contiguous run
into one interval node in a single split--merge splice. Our central result is a
tight characterisation of this structure by its \emph{fragmentation} $A$---the
number of index pieces, namely live handles plus maximal same-digest runs. The map
uses $\Theta(A)$ space, resolves a point in $\bigoh(\log A)$, reports a $k$-handle
range in $\bigoh(\log A + k)$, and updates in $\bigoh((a/B + s)\log A)$ over $a$
appends and $s$ supersedes; matching $\Omega(A)$ space and $\Omega(\log A + k)$ query
lower bounds hold for \emph{any} reference-stable, ordered, bounded-working-set
structure, and an adversary forces $A=\Theta(s)$. The index is thus
instance-optimal---$o(n)$ when appends dominate, $\Theta(s)$ under adversarial edits,
and then no structure does better. Resolution is total, snapshot-consistent, and
anomaly-free; the working set is independent of history length; and the design
transfers to external memory at $\bigoh(\log_\beta A)$ I/Os. A reference
implementation and a reproducible study to $10^{6}$ records, including adversarial and
distribution-free workloads and a copy-on-write ordered-map baseline, trace $A$
directly and confirm the bounds.
\end{abstract}

\section{Introduction}
\label{sec:intro}

Many systems are organised around a single growing sequence of immutable
records that is appended to far more often than it is read in full: the edit
history of an interactive session, the event log of an event-sourced
application, the revision stream of a version-controlled working set, the chunk
sequence of an incremental indexer, the tuple stream of a continuous analytic
query. The sequence quickly outgrows fast memory, so it is \emph{tiered}. A
bounded \emph{hot} stratum holds the recent records fully materialised; aged
records are \emph{folded} in contiguous blocks into colder strata, where a block
is replaced by a compact \emph{digest} that summarises it. Tiering is what keeps
the materialised working set bounded as the history grows without limit, and it
is the same idea that underlies log-structured storage~\citep{ONeil1996lsm}.

The records are not independent. Each carries a stable \emph{handle}, and
records reference earlier handles: an edit refers to the region it revises, an
event to the entity it mutates, a derived chunk to its source. A consumer
resolves a handle to obtain the record's current version, and assembles a
bounded \emph{window} of related records around a query. The difficulty is that
references cross tier boundaries. A handle minted while its target was hot is
followed later, after the target has been folded into a cold digest; a handle is
edited after the block holding it was already summarised; a window is assembled
while a fold is in flight. At each of these boundaries the obvious tiered design
fails in a way that is easy to miss because every individual tier is internally
consistent---the defect lives only in the transition between tiers.

We make the failure precise as three \emph{cross-tier reference anomalies}.
A \emph{dangling} resolution returns nothing for a handle that was appended and
never deleted, because folding reclaimed its addressing entry. A \emph{stale}
resolution returns a superseded version, because no tier holds the authoritative
one. A \emph{snapshot-skew} window mixes records read before and after a fold and
is therefore not a consistent view of any single instant. \Cref{sec:empirical}
quantifies them: with a hot stratum of a few hundred records over a history of
tens of thousands, a tiered log leaves under two percent of the history
addressable and mis-resolves roughly a quarter of references drawn from a
realistic access mixture.

This paper presents the \cascade{}, a tiered append structure for which
resolution is provably free of all three anomalies while the working set stays
bounded. Three ingredients combine:
\begin{itemize}
\item a \emph{versioned interval index} (\Cref{sec:vii})---one persistent search
  tree, shared across snapshots by path copying~\citep{DriscollSarnakSleatorTarjan1989,SarnakTarjan1986}---which spans every
  stratum and is the single authority on where each handle's current version
  lives, so resolution never depends on the order in which tiers are probed;
\item a \emph{coalescing interval map} (\Cref{sec:fold})---folding a contiguous
  run collapses it to a single interval node in one $\bigoh(\log n)$ splice, so
  compaction is amortised $\bigoh(1)$ per appended record and the index is
  sublinear, while a folded handle still resolves through its digest; and
\item \emph{snapshot tokens} (\Cref{sec:snap})---immutable roots that serve
  reads as of any past instant, making windows consistent by construction.
\end{itemize}
Windows are chosen (\Cref{sec:window}) by diffusing relevance from the query over
the reference graph with a personalised random walk~\citep{BrinPage1998,Haveliwala2003}
and then solving a standard budgeted maximum-coverage instance over the
candidates~\citep{NemhauserWolseyFisher1978,KhullerMossNaor1999}. The selection
layer applies known algorithms; our concern is the structure beneath it that
keeps the candidates resolvable, snapshot-consistent, and cheap to gather.

\paragraph{Contributions.}
The data-structural result is a tight characterisation of reference-stable
compaction by its \emph{fragmentation}.
\begin{enumerate}[label=(\roman*)]
\item The \emph{coalescing persistent interval map} (\Cref{sec:structure}): a
  persistent ordered map whose entries are intervals, where a folded contiguous run
  of handles is one node installed by a single split--merge splice, with a precisely
  specified update sequence---$\textsc{append}$, $\textsc{materialise}$,
  $\textsc{summarise}$, $\textsc{stab}$.
\item A \emph{tight fragmentation trade-off} (\Cref{thm:frag},
  \Cref{cor:instance})---the central result. With $A$ the number of index pieces
  (live handles plus same-digest runs), the map uses $\Theta(A)$ space, resolves in
  $\bigoh(\log A)$, reports a $k$-range in $\bigoh(\log A + k)$, and updates in
  $\bigoh((a/B+s)\log A)$; matching lower bounds force $\Omega(A)$ space and
  $\Omega(\log A + k)$ query for \emph{any} reference-stable, ordered,
  bounded-working-set structure, and an adversary drives $A=\Theta(s)$. The
  index is thus instance-optimal: $o(n)$ on append-dominated histories and
  $\Theta(s)$ under adversarial edits---and then so is any competitor's.
\item Supporting bounds: a space bound that versions cannot be lossily compressed
  below $\Omega(n\log t)$ bits (\Cref{thm:lower}), an ordered-range separation from
  hash indexes (Proposition~\ref{prop:order}), and a fully specified external-memory
  realisation with $\bigoh(\log_\beta A)$-I/O resolve and fold and
  $\bigoh(\log_\beta A + k/\beta)$ range (\Cref{thm:io}).
\item Correctness, conceded to be immediate from the single-authority invariant:
  resolution is total (\Cref{thm:total}), snapshot-consistent (\Cref{thm:linear}),
  and anomaly-free (\Cref{thm:anomaly}) under an explicit partial-persistence model;
  the work is the cost of maintaining the invariant, above.
\item A reference implementation, an oracle-checked anomaly model, and a reproducible
  study to $10^{6}$ records with adversarial and distribution-free
  workloads against a copy-on-write ordered-map (multiversion-B-tree) baseline
  (\Cref{sec:empirical}). A budgeted-window layer (\Cref{sec:window}) is included as
  an application; its guarantees are standard maximum coverage and not part of the
  core claim.
\end{enumerate}

\section{Model}
\label{sec:model}

\paragraph{Records, handles, versions.}
The structure stores a sequence of records appended over time. Appending a
record mints a fresh handle from a monotonically increasing counter; the $k$-th
appended record has handle $k$. A record may be \emph{superseded} by a later
version of the same handle, modelling an in-place edit, so a handle denotes a
logical record whose content evolves through versions $0,1,2,\dots$ Each handle
$h$ carries a set of out-references $R(h)$ to earlier handles; the directed graph
$G=(V,E)$ with $E=\{(h,h')\,:\,h'\in R(h)\}$ is the \emph{reference graph}.

\paragraph{Strata and folding.}
Records occupy a sequence of strata $T_0,T_1,\dots$ with the hot stratum $T_0$
holding the recent records fully materialised. When the number of materialised
records exceeds a capacity $C$, the oldest are demoted; once a block of $B$
demoted records accumulates it is \emph{folded} into a \emph{digest}, a
fixed-footprint object that summarises the block. A cost $c(h)>0$ (think bytes or
tokens) is charged for each materialised version; a digest costs $\bigoh(1)$
irrespective of $B$.

\paragraph{Operations.}
A structure supports $\append(p,c,R)\to h$, $\supersede(h,p,c)\to v$,
$\resolve(h,\tau)$ returning the version of $h$ live at instant $\tau$ (or the
current instant), and $\window(q,\beta)$ returning a budgeted set of records
relevant to a query handle $q$ within cost budget $\beta$.

\paragraph{Persistence model.}
\label{sec:persistence}
Updates ($\append$, $\supersede$, fold) apply only to the newest version and
produce a linear sequence of versions $D_0,D_1,\dots$; we index them by an
integer \emph{epoch}, the count of updates applied. The structure is therefore
\emph{partially persistent} in the sense of \citet{DriscollSarnakSleatorTarjan1989}:
any past epoch is read-only and queryable, but the history does not branch---there
is no operation that forks an alternate timeline. A \emph{snapshot token} names
an epoch; $\resolve(h,\tau)$ with a token reads version $D_\tau$. We make the
retention policy explicit because it governs both correctness and space:
\begin{itemize}
\item \emph{Ephemeral mode} retains only the current epoch. Past epochs are not
  queryable, and a version superseded by a later update becomes unreachable and
  is reclaimed; the space bound (\Cref{thm:complexity}) is stated for this mode.
\item \emph{Retained mode} keeps a caller-chosen set of tokens; the epochs they
  name remain queryable, and the path-copied state they reach is kept alive and
  \emph{counted} in the space bound. Reclamation is by reachability from the
  current root and the retained tokens (reference counting suffices, as the
  version graph is acyclic).
\end{itemize}
A digest is reachable, hence retained, exactly while some live or retained epoch
maps a handle to it; once every handle of a digest has been superseded and no
retained token predates those edits, the digest is reclaimed.
``Resolve against any past instant'' thus means any \emph{retained} epoch, not an
arbitrary one---no claim is made about epochs whose tokens were not kept.

\paragraph{Correct resolution.}
Let $\mathrm{ver}^\star_\tau(h)$ and $\mathrm{pay}^\star_\tau(h)$ be the true
version and payload of $h$ at instant $\tau$. A summarised resolution returns
$\big(h,\ \mathrm{ver}^\star_\tau(h),\ \sigma(\mathrm{pay}^\star_\tau(h))\big)$,
where $\sigma$ is the digest's summary map; it is \emph{version-faithful} always and
\emph{payload-faithful} exactly when $\sigma$ is the identity. A resolution
$\rho=\resolve(h,\tau)$ is \emph{correct} when $h$ exists at $\tau$ and $\rho$
reports the handle $h$ with version $\mathrm{ver}^\star_\tau(h)$: correctness and the
anomaly classification below concern identity and version only, a guarantee strictly
weaker than payload-faithful resolution, which holds iff $\sigma=\mathrm{id}$.
We classify an incorrect resolution as one of the following \emph{cross-tier
anomalies}.

\begin{definition}[Cross-tier anomalies]
\label{def:anomaly}
A resolution of an existing handle $h$ at instant $\tau$ is
\emph{dangling} if it returns nothing;
\emph{stale} if it returns a version $v<\mathrm{ver}^\star_\tau(h)$;
\emph{corrupt} if it returns a different handle or, for a materialised result, a
payload other than $\mathrm{pay}^\star_\tau(h)$.
A window is \emph{snapshot-skewed} if its members are resolved against two
instants separated by a fold, so that no single instant explains the window.
\end{definition}

These anomalies are the formal content of the informal claim that ``a hidden
defect appears in the transition between layers'': each is triggered only by a
reference whose target crosses a stratum boundary during folding.

\section{The \cascade{}}
\label{sec:structure}

\Cref{fig:arch} shows the structure. A single index spans all strata; folding
moves payloads downward but leaves the index total; immutable roots capture
snapshots.

\begin{figure}[t]
\centering
\begin{tikzpicture}[font=\small,>=Stealth]
  \node[draw, fill=gray!8, rounded corners, minimum width=23mm, minimum height=15mm,
        align=center] (map) at (-3.7,0) {versioned\\ interval map\\ \itshape(persistent)};
  \node[font=\scriptsize, left=5mm of map, align=center] (q)
        {$\resolve(h,\tau)$\\[1pt]$\window(q,\beta)$};
  \draw[->,thick] (q) -- (map.west);

  \node[draw, rounded corners, fill=blue!8, minimum width=50mm, minimum height=7.5mm] (hot) at (2.7,0.95) {};
  \foreach \i/\x in {1/-21,2/-13.5,3/-6,4/1.5,5/9,6/16.5}{
     \node[draw, fill=blue!20, minimum width=5mm, minimum height=5mm] (r\i) at ($(hot.center)+(\x mm,0)$) {};}
  \node[font=\scriptsize, anchor=north] at (hot.south) {hot stratum: live singletons (materialised)};
  \node[font=\scriptsize, right=1mm of hot] {$\leftarrow$ append};

  \node[draw, rounded corners, fill=orange!10, minimum width=50mm, minimum height=7.5mm] (cold) at (2.7,-1.25) {};
  \node[draw, fill=orange!22, minimum width=16mm, minimum height=5mm] (d1) at ($(cold.center)+(-11mm,0)$) {\scriptsize digest};
  \node[draw, fill=orange!22, minimum width=16mm, minimum height=5mm] (d2) at ($(cold.center)+(12mm,0)$) {\scriptsize digest};
  \node[font=\scriptsize, anchor=north] at (cold.south) {cold: each summarised run is one node $\to$ one digest};

  \draw[->, blue!60!black] (map.east) to[out=22,in=180] node[above,font=\scriptsize,pos=0.6]{stab} (r2.west);
  \draw[->, orange!70!black, dashed] (map.east) to[out=-22,in=180] node[below,font=\scriptsize,pos=0.6]{stab} (d1.west);

  \draw[decorate,decoration={brace,amplitude=4pt}]
     ($(r2.north west)+(0,4.5mm)$) -- node[above,font=\scriptsize,yshift=1pt]{$\window(q,\beta)$} ($(r4.north east)+(0,4.5mm)$);
  \draw[gray!55] (r2.north west) -- ($(r2.north west)+(0,4.5mm)$);
  \draw[gray!55] (r4.north east) -- ($(r4.north east)+(0,4.5mm)$);
\end{tikzpicture}
\caption{The \cascade{}. Appends enter the hot stratum $T_0$; aged blocks fold
into digests in colder strata. The versioned interval map spans the working set
and the digests and resolves every handle to its live version---a materialised hot
record (solid) or a digest reached by stabbing its summarised interval node
(dashed). Immutable roots of the map serve snapshot-consistent windows.}
\label{fig:arch}
\end{figure}
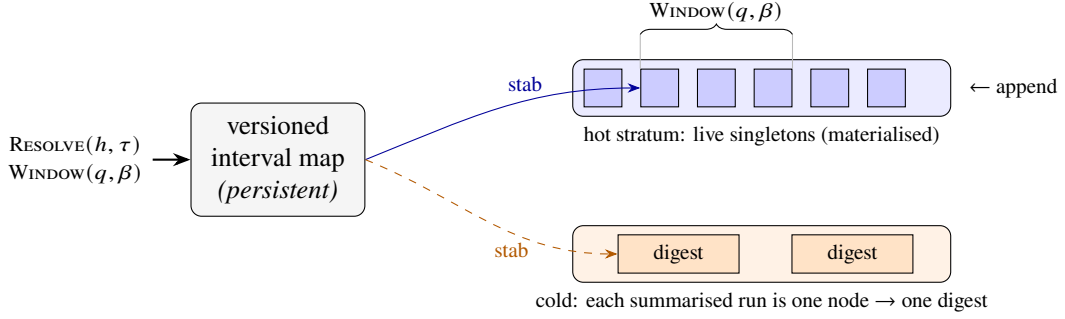

\subsection{The versioned interval map}
\label{sec:vii}

The index is a persistent ordered map over handles whose entries are
\emph{intervals}, not individual keys. A live handle is a singleton interval
$[h,h]$ carrying its current version and payload; a maximal contiguous run of
summarised handles is a \emph{single} interval node $[a,b]$ that names a digest
and stands for every handle it spans. We call it the versioned interval map
(\VII{}). It supports four operations:
\begin{itemize}
\item $\textsc{stab}(h)$ -- locate the interval containing $h$ ($\bigoh(\log n)$);
\item $\textsc{append}(h)$ -- add a singleton for a new maximal handle;
\item $\textsc{materialise}(h)$ -- make $h$ a live singleton, \emph{splitting} the
  run that holds it into the parts below and above $h$; and
\item $\textsc{summarise}(a,b,d)$ -- replace the live singletons of a contiguous
  range by one interval node pointing at digest $d$.
\end{itemize}
The last is a genuine interval (range) update and the first a point-stabbing
query, which is what makes the structure an interval object rather than a
per-handle dictionary: a folded run of any length occupies one node, so the map
holds one node per live handle plus one per summarised run---sublinear in $n$
when folding dominates (\Cref{sec:empirical}).

The map is \emph{persistent} by path copying: every update allocates
$\bigoh(\log n)$ fresh nodes along one root-to-leaf path and shares all other
subtrees with the previous version, so the previous root remains a valid,
immutable view~\citep{DriscollSarnakSleatorTarjan1989,SarnakTarjan1986}. The
reference implementation uses a treap with priorities a fixed hash of the
interval's lower endpoint, giving expected height $\bigoh(\log n)$ with
reproducible shape~\citep{SeidelAragon1996treaps}; any persistent balanced search
tree yields worst-case $\bigoh(\log n)$ and may be substituted. Because the
\VII{} is the sole record of each handle's live location, resolution consults one
authority rather than probing strata in order; this is the source of the absence
of stale and snapshot-skew anomalies.

\subsection{Folding by coalescing runs}
\label{sec:fold}

When records age out of the working set they are demoted; once a block of $B$ has
accumulated it is folded. Folding a maximal contiguous run $[a,b]$ of demoted
singletons is a single $\textsc{summarise}$: the run is split out of the map and
one interval node spliced back in its place, in $\bigoh(\log n)$ \emph{regardless
of} $b-a+1$ (\Cref{alg:fold}). An append-dominated stream demotes contiguous
handles, so a whole block folds in one splice; edits that re-materialise old
handles break runs, and each such fragmentation is one extra splice charged to the
edit (\Cref{thm:complexity}). The bulky payloads leave the working set; what is
retained, per folded handle, is its exact version inside the digest.

\begin{algorithm}[t]
\caption{$\fold$ of a demoted block: one splice per contiguous run}
\label{alg:fold}
\begin{algorithmic}[1]
\Require demoted block $\{(h_i,v_i)\}$; current root $r$
\State drop any $h_i$ re-edited since demotion (it is live again); sort the rest
\For{each maximal contiguous run $[a,b]$ of the survivors}
  \State $d \gets \{\,h \mapsto (v_h,\ \textsc{summary}(\text{payload}_h))\ :\ h\in[a,b]\,\}$ \label{line:digest}
  \State $r \gets \VII.\textsc{summarise}(r,a,b,d)$ \Comment{one $\bigoh(\log n)$ splice; $b{-}a{+}1$ nodes $\to 1$} \label{line:splice}
  \State release the run's payloads from the working set
\EndFor
\State \Return updated root $r$
\end{algorithmic}
\end{algorithm}

Line~\ref{line:splice} is the mechanism: a run collapses to one node in a single
split/merge, so the index work to fold a block is one splice per contiguous run,
not one update per handle. A design that instead rewrites each handle's entry---a
hash spine (\Cref{sec:impl}), or any per-key forwarding scheme---keeps a linear
index and does $\Theta(B)$ index work per block; \Cref{thm:complexity} charges the
difference.

\paragraph{Digest semantics.}
A digest stores, for each handle it covers, that handle's \emph{exact} version and
a faithful summary of its payload (line~\ref{line:digest}); the summary is a
function of the folded bytes, not an arbitrary value. Resolution therefore returns
the correct handle and version for a summarised record always, and the summary in
place of the original bytes---the only approximation folding introduces, and one
an application opts into by choosing $\textsc{summary}$ (the identity map makes
the digest lossless). A per-digest membership filter~\citep{Bloom1970} only
accelerates negative probes across digests; it never participates in correctness,
which is decided by the map, so a filter false positive costs one wasted digest
lookup and never a wrong answer.

\subsection{Snapshots and resolution}
\label{sec:snap}

A snapshot is taken by retaining the current immutable root as a token; holding a
token keeps that version of the \VII{} alive. \Cref{alg:resolve} resolves a
handle against the current root or any token in $\bigoh(\log n)$: it locates the
entry and returns the materialised payload (live) or the digest-backed view
(summarised). Resolving every member of a window against one token yields a view
of a single instant.

\begin{algorithm}[t]
\caption{$\resolve(h,\tau)$}
\label{alg:resolve}
\begin{algorithmic}[1]
\Require handle $h$; instant token $\tau$ (default: current root)
\State $I \gets \VII.\textsc{stab}(\tau, h)$ \Comment{the interval containing $h$}
\If{$I=\varnothing$} \Return $\varnothing$ \Comment{$h$ not present at $\tau$} \EndIf
\If{$I$ is a live singleton} \Return $\langle h, I.\text{version}, I.\text{payload}, \text{materialised}\rangle$ \EndIf
\State $(v_h, s_h) \gets \text{digest}[\,I.\text{digest}\,][h]$
\State \Return $\langle h, v_h, s_h, \text{summarised}\rangle$
\end{algorithmic}
\end{algorithm}

Appending splices a fresh maximal singleton in one pass; superseding makes the
handle a live singleton, so the new version is authoritative the instant it is
written, regardless of any older copy lingering in a digest.

\subsection{Cross-tier range reporting}
\label{sec:range}

Reporting all handles in a key range that straddles the working set and several
digests is an in-order traversal of the \VII{} between two endpoints, expanding
each summarised interval it meets, and costs $\bigoh(\log n + k)$ for $k$ reported
handles. A per-stratum design would search each stratum's index separately; the
classical way to recover a single logarithmic search across ordered layers is
fractional cascading~\citep{ChazelleGuibas1986}, and layered probabilistic search
lists are an alternative~\citep{Pugh1990skiplists}. The \cascade{} attains the
$\bigoh(\log n + k)$ bound directly because its one ordered map already spans the
working set and the digests. A hash index over handles cannot report a key range
at all without enumerating its whole contents.

\subsection{Application: budgeted windows}
\label{sec:window}

This subsection is an application layer over the structure, not part of the core
data-structural contribution; its guarantees are inherited from budgeted maximum
coverage and a reader interested only in the index may skip it.

A window around a query handle $q$ is chosen in two stages. First, relevance is
diffused from $q$ over the reference graph by a personalised random walk with
restart on $q$---a personalised PageRank whose stationary distribution
concentrates on the citation neighbourhood of
$q$~\citep{BrinPage1998,Haveliwala2003}. Second, among the candidates a subset
$S$ is chosen to maximise the covered relevance
$f(S)=\sum_{v\in\cover(S)} w(v)$, where $\cover(u)$ is $u$ together with the
handles it references and $w$ is the diffused weight, subject to
$\sum_{u\in S}c(u)\le\beta$. Since $f$ is a weighted coverage function---monotone
and submodular---the task is exactly \emph{budgeted maximum
coverage}~\citep{KhullerMossNaor1999}. Two algorithms apply. The linear-time
cost--benefit greedy of \Cref{alg:window}, which returns the better of its
density-greedy run and the best single feasible record, attains $\tfrac12(1-1/e)$
of the optimum; the partial-enumeration of \Cref{alg:window-partial}, which seeds
the greedy from every feasible set of at most three records, attains $1-1/e$. The
prefix heuristic---taking records in decreasing relevance until the budget is
spent---ignores both cost-effectiveness and the overlap between revealed contexts,
and has no constant-factor guarantee (Proposition~\ref{prop:adversarial}).

The guarantee is over the candidate pool $\mathcal{C}(q)$ the diffusion produces and
the cost--coverage instance on it; it does \emph{not} bound the loss from candidate
generation---from truncating the personalised-PageRank neighbourhood, or from an
optimal record lying outside $\mathcal{C}(q)$. The cost $c(u)$ is a single static
weight standing for a record's logical size, its materialisation cost, and the
digest lookup a summarised record incurs; budgeted coverage treats it as fixed,
while the live structure has tier-dependent access costs, so the guarantee is exact
for the abstracted instance and indicative for the system. This selection layer is
standard---a reduction to budgeted maximum coverage with its known
guarantees~\citep{NemhauserWolseyFisher1978,KhullerMossNaor1999}; the contribution
here is the structure beneath it that makes candidates resolvable, snapshot-consistent,
and cheap to gather (\Cref{sec:range}), and Proposition~\ref{prop:adversarial} is the
textbook knapsack--coverage gap, included only to show the prefix heuristic is unsafe.

\begin{algorithm}[t]
\caption{$\window(q,\beta)$: simple greedy, $\tfrac12(1-1/e)$}
\label{alg:window}
\begin{algorithmic}[1]
\State $w \gets \textsc{PersonalizedPageRank}(G,\,\text{seed}=q)$
\State $\text{candidates}\gets$ resolvable handles near $q$; \; $\cover(u)\gets\{u\}\cup R(u)$
\State $S\gets\varnothing$; \; $U\gets\varnothing$; \; $\text{spent}\gets0$
\While{some feasible $u$ remains}
  \State $u^\star \gets \arg\max_{u:\ \text{spent}+c(u)\le\beta}\ \dfrac{\sum_{v\in\cover(u)\setminus U} w(v)}{c(u)}$
  \State $S\gets S\cup\{u^\star\}$; \; $U\gets U\cup\cover(u^\star)$; \; $\text{spent}\mathrel{+}=c(u^\star)$
\EndWhile
\State \Return $\arg\max\{\,f(S),\ \max_{u:\,c(u)\le\beta} f(\{u\})\,\}$, materialised via \resolve
\end{algorithmic}
\end{algorithm}

\begin{algorithm}[t]
\caption{$\window^{+}(q,\beta)$: partial enumeration, $1-1/e$}
\label{alg:window-partial}
\begin{algorithmic}[1]
\State compute $w$, $\text{candidates}$, $\cover$ as in \Cref{alg:window}
\State $S^\star\gets\varnothing$
\For{each feasible seed $H\subseteq\text{candidates}$ with $|H|\le 3$}
  \State $S\gets H$
  \While{some feasible $u\notin S$ remains}
     \State $S\gets S\cup\{\arg\max_{u}\ \tfrac{1}{c(u)}\!\sum_{v\in\cover(u)\setminus\cover(S)} w(v)\}$ \Comment{$\text{spent}+c(u)\le\beta$}
  \EndWhile
  \If{$f(S)>f(S^\star)$} \State $S^\star\gets S$ \EndIf
\EndFor
\State \Return $S^\star$, materialised via \resolve
\end{algorithmic}
\end{algorithm}

\section{Analysis}
\label{sec:analysis}

Throughout, $n$ is the number of handles ever appended.

\begin{theorem}[Totality]
\label{thm:total}
For every handle $h$ and every instant $\tau$ at or after $h$'s append,
$\resolve(h,\tau)$ returns a defined view with identity $h$.
\end{theorem}
\begin{proof}
We show the invariant that the live root contains an entry for every appended
handle. Appending inserts a live singleton. Superseding makes the handle a live
singleton, preserving presence. Folding (\Cref{alg:fold}, line~\ref{line:splice})
replaces a run of live singletons by one summarised interval node naming a digest
that retains every handle's version; the handles stay covered and the digest is
never discarded. No operation removes a handle from the map. Hence the live root,
and by persistence every root derived from it at an instant
$\tau\ge\text{append}(h)$, covers $h$, on which \Cref{alg:resolve} returns a
defined view with identity $h$.
\end{proof}

\begin{theorem}[Snapshot consistency]
\label{thm:linear}
Fix a retained token naming epoch $\tau$. Every resolution against that token
returns the value prescribed by the canonical sequential execution truncated
after its first $\tau$ updates. Consequently a window all of whose members are
resolved against one token is a consistent snapshot of epoch $\tau$ and is never
snapshot-skewed.
\end{theorem}
\begin{proof}
A token is an immutable root produced by path copying: each update derives a new
root from the previous one without mutating any existing
node~\citep{DriscollSarnakSleatorTarjan1989}. Hence the root at epoch $\tau$ is,
node for node, the tree obtained by applying exactly the first $\tau$ updates of
the linear history (\Cref{sec:persistence}); reading any set of handles against it
returns the state after update $\tau$ and no later update. A window built from one
token therefore cannot mix pre- and post-fold states.
\end{proof}

\begin{remark}[Concurrent readers]
\label{rem:linear}
The structure as specified is sequential, but its immutability gives the
single-writer, many-reader model for free, and that model is exactly the read path
of multiversion concurrency control~\citep{BernsteinGoodman1983}: a writer serialises
updates and publishes each new root with one atomic pointer store, while readers
load the pointer once and traverse the immutable tree it names. Readers never block,
never coordinate, and each observes a consistent snapshot; every operation
linearises at its publish (writer) or load (reader), so the object is linearizable
in the sense of~\citet{HerlihyWing1990}. This is the regime the motivating
applications read in---snapshot-isolation queries and time-travel reads run against
exactly such immutable versions---and the working set keeps a writer's update local.
A multi-writer, lock-free design, where concurrent writers merge into one version
order, is genuinely open and is the main limitation of the present primitive.
\end{remark}

\begin{theorem}[Anomaly-freedom]
\label{thm:anomaly}
No execution of the \cascade{} produces a dangling, stale, corrupt, or
snapshot-skewed resolution.
\end{theorem}
\begin{proof}
Dangling is excluded by \Cref{thm:total}. For stale and corrupt, fix a token at
instant $\tau$ and a handle $h$; in that root $h$ carries the version written by
the most recent append or supersede of $h$ at or before $\tau$, since both write
$h$ as a live singleton and folding copies the version verbatim into the digest.
\Cref{alg:resolve} returns exactly that version and the handle $h$, with a payload
that is the materialised value when live and the faithful digest summary when
summarised; hence the version equals $\mathrm{ver}^\star_\tau(h)$, the
identity is $h$, and a materialised payload equals $\mathrm{pay}^\star_\tau(h)$.
So no resolution is stale or corrupt. Snapshot-skew is excluded by
\Cref{thm:linear}.
\end{proof}

The three guarantees above are by construction: a single authoritative map that
never discards a version makes them near-immediate, and we claim no subtlety in
the invariant proofs. The substance is the \emph{cost} at which they hold
(\Cref{thm:complexity}) and the sense in which that cost is forced. How cheaply
can a folded handle's \emph{version} be kept? \Cref{thm:lower} answers this and is
best read as a space statement that no lossy payload summary can evade: a structure
free to summarise payloads arbitrarily must still retain $\Omega(n\log t)$ bits to
recover exact versions across folds. It is a \emph{supporting} bound on the cold
digests---it shows versions cannot be lossily compressed---and is deliberately not the
structural claim: the bound that pins down the \emph{index} and establishes optimality
of the interval representation is the fragmentation trade-off (\Cref{thm:frag}), and
the cost of \emph{ordered} access is \Cref{prop:order}. \Cref{thm:lower} rules out only
the economy of summarising versions along with the payloads they annotate.

\begin{theorem}[Space lower bound for exact versions across folds]
\label{thm:lower}
Work in the bit-probe model: a data structure is a string of memory bits, and an
operation reads and writes bits at addresses of its choosing. Consider a block of
$B$ handles, each of which the workload may independently leave unedited or
supersede to one of $t-1$ distinguishable later versions ($t\ge 2$). Let $S$ be the
memory bits the structure retains after the block is folded and reads while
answering a later $\resolve$ of a block handle for its \emph{version}. If every such
$\resolve$ returns the handle's current version for all $t^{B}$ edit
patterns---whatever \textup{(}possibly lossy\textup{)} summary of the payloads the
structure also keeps---then $|S|\ge B\log_2 t$; over $m$ disjoint folded blocks the
retained version state is $\Omega(n\log t)$ bits.
\end{theorem}
\begin{proof}
Fix the decoder: the version $\resolve(h)$ returns after the fold is a
deterministic function of the bits in $S$ and the (workload-independent) query
address; any payload summary the structure keeps is irrelevant to the version it
reports. Run the workload for an edit pattern $x\in\{0,\dots,t-1\}^{B}$, fold the
block, and read off the $B$ version answers; by correctness they reveal the current
version of every handle, hence $x$. So $x\mapsto S(x)$ is injective on its $t^{B}$
inputs, its range has at least $t^{B}$ values, and $|S|\ge\log_2 t^{B}=B\log_2 t$.
The blocks are disjoint, so the bounds add, giving $\Omega(mB\log t)=\Omega(n\log t)$.
\end{proof}

In the \cascade{} this irreducible information lives in the digests: a digest
records each handle's version in $\bigoh(\log t)$ bits, so the digests occupy
$\bigoh(n\log t)$ bits, meeting \Cref{thm:lower}. The interval map adds only
$\bigoh(1)$ words per live handle and per summarised run, which is $o(n)$ when
folding dominates (\Cref{thm:complexity})---so the bound is met by the cold
digests, not by a per-handle index. We claim no more: \Cref{thm:lower} bounds the
bits any structure must keep recoverable, and the digests realise it up to the
encoding of a version; a succinct or compressed version map could lower the
constant at a time cost, which we do not rule out. What the bound \emph{does}
exclude is the family that keeps $o(\log t)$ bits per folded handle---in
particular the fold that discards versions, which restarts a re-edited handle at
version one (a stale read) or loses it (a dangling read), the collapse measured in
\Cref{sec:empirical}.

The version bound says nothing about \emph{query} cost or about the order the
\cascade{} maintains; that cost has its own, separate lower bound, and it is where
the \cascade{} genuinely parts from a hash.

\begin{proposition}[Ordered range is not free]
\label{prop:order}
A hash index over $n$ handles---expected $\bigoh(1)$ point resolution, no order---has
no range-reporting algorithm faster than testing every stored handle, so it answers
a range query in $\Theta(n)$ time independent of the output size $k$
(\Cref{sec:empirical} confirms this for \textsf{HashSpine}). Sub-linear range and
predecessor search require an ordered layout and even then cannot be done in
$\bigoh(1)$ cell probes per query within the space the \cascade{}
uses~\citep{BeameFich2002,PatrascuThorup2006}. The \cascade{} reports a range of
$k$ handles in $\bigoh(\log n + k)$ (\Cref{sec:range}) and resolves a point in
$\bigoh(\log n)$; the logarithmic factor over a hash is the price of the ordered,
snapshot-consistent access a hash cannot provide at all.
\end{proposition}

\begin{theorem}[Total work and space]
\label{thm:complexity}
Run any sequence of operations on an initially empty \cascade{} with hot capacity
$C$ and fold block $B$, comprising $a$ appends and $s$ supersedes and the folds
they trigger; let $n=a$. Then:
\begin{enumerate}[label=\textup{(\alph*)},leftmargin=2.2em,itemsep=1pt]
\item \emph{Per operation.} Append, supersede, and resolve each cost
  $\bigoh(\log n)$ expected and allocate $\bigoh(\log n)$ index nodes; range
  reporting costs $\bigoh(\log n + k)$. The bounds are expected over the treap's
  hashed priorities and become worst-case if the \VII{} is realised with a
  deterministic persistent balanced tree (\Cref{sec:vii}); all bounds below are
  stated this way.
\item \emph{Compaction.} A fold splices one interval node per maximal contiguous
  run of demoted handles, $\bigoh(\log n)$ each, \emph{independent of the run
  length}. Compacting the appended records uses $\bigoh(a/B)$ splices---amortised
  $\bigoh(1)$ index work per appended record once $B=\Omega(\log n)$---while each
  supersede fragments at most one run and triggers at most one further splice,
  $\bigoh(\log n)$ charged to that supersede. Payload compaction is $\bigoh(1)$
  amortised per record.
\item \emph{Working set.} At most $C+B$ records are materialised at any instant,
  so the working set is $\bigoh((C+B)c_{\max})$, independent of $n$.
\item \emph{Index.} The map holds $\bigoh(L+\Phi)$ nodes, $L\le C+B$ live and
  $\Phi$ summarised runs with $\Phi\le 1+s$; an append-dominated history has
  $\Phi=\bigoh(a/B)$, an $o(n)$ index.
\item \emph{Cold store and reclamation.} Digests hold $\bigoh(n\log t)$ bits
  (\Cref{thm:lower}); an entry is reclaimed when its handle is re-materialised and
  a digest when its last entry is, by reference counting in $\bigoh(1)$ each, so
  in ephemeral mode the live digest content is the number of currently summarised
  handles, $\le n$.
\item \emph{Retained snapshots.} Retained roots share every subtree they have in
  common, so the overhead is the number of \emph{distinct} nodes path-copied by
  updates within the kept epochs' reachability---$\bigoh(U_{\mathrm{ret}}\log n)$
  words for $U_{\mathrm{ret}}$ such updates---not $R$ copies of the tree. This
  sharing is why retaining $R$ tokens raises the index to a small constant per
  handle rather than by a factor of $R$ (\Cref{fig:snapshot} measures it).
\end{enumerate}
The total time over the sequence is $\bigoh((a+s)\log n)$.
\end{theorem}
\begin{proof}
(a) $\textsc{append}$, $\textsc{materialise}$ (a split and a merge) and
$\textsc{stab}$ each follow one root-to-leaf path of the treap, of expected length
$\bigoh(\log n)$, copying only that path; range reporting visits $\bigoh(\log n)$
boundary nodes and $\bigoh(k)$ interior ones.
(b) $\textsc{summarise}(a,b,d)$ is two splits and two merges, $\bigoh(\log n)$
regardless of $b-a+1$, replacing the run's nodes by one. A handle never
re-materialised since its append is, when demoted, part of the still-untouched
tail of the append order and hence contiguous with its neighbours, so compacting
the $a$ appends uses $\bigoh(a/B)$ runs of one splice each. A supersede makes its
handle a live singleton, splitting at most one summarised run (one extra node, and
one extra splice when that fragment later folds); both are $\bigoh(\log n)$ charged
to the supersede. A digest over a run of length $B'$ is built in $\bigoh(B')$ and
charged $\bigoh(1)$ to each of its members~\citep{Tarjan1985amortized}.
(c) Demotion fires whenever more than $C$ records are active and a block of $B$
folds together, so the materialised count stays in $[\,C,\,C+B\,]$.
(d) Immediate from the interval representation; each supersede creates at most one
new run boundary, so $\Phi\le 1+s$.
(e) A re-materialisation deletes the handle's digest entry; versions form a
forest, so reference counting frees entries and digests in $\bigoh(1)$ each.
(f) A path-copied node is freed once unreachable from the live root and every
retained token. Summing (a) over the $a+s$ operations and adding (b) gives total
time $\bigoh((a+s)\log n)$.
\end{proof}

\begin{remark}[Explicit potential for the amortised fold]
\label{rem:potential}
The charge in \Cref{thm:complexity}(b) is witnessed by the potential
$\Phi=\kappa\,(\log n)\,p$, where $p$ is the number of demoted handles not yet
folded and $\kappa$ exceeds the per-splice constant. Each append or supersede demotes
at most one handle, raising $\Phi$ by $\kappa\log n$, folded into its
$\bigoh(\log n)$ amortised cost; folding a contiguous run of $\ell\ge 1$ handles
does one $\bigoh(\log n)$ splice (a $\textsc{split}$, $\textsc{split}$, $\textsc{merge}$,
$\textsc{merge}$ that copies two $\bigoh(\log n)$ boundary paths and discards the
interior subtree of the run) while releasing $\kappa\ell\log n\ge\kappa\log n$
potential, so its amortised cost is non-positive. A supersede that fragments a
summarised run pays the split's $\bigoh(\log n)$ directly; the fragment's eventual
re-fold is covered by the same potential when its handle re-demotes. Digest
construction and the copied digest pointer are $\bigoh(1)$ per folded handle, and
retained tokens only delay node reclamation, not the charge. The sum is
$\bigoh((a+s)\log n)$.
\end{remark}

The bounds so far are stated in $n$. The sharper and, we think, central statement
is in terms of how \emph{fragmented} the folded history is, because that is the
quantity an adversary controls and the quantity against which the structure is
optimal.

\begin{definition}[Fragmentation]
\label{def:frag}
At an instant let $\ell$ be the number of live (materialised) handles and $\rho$ the
number of maximal runs of consecutive handles that map to one common digest. The
\emph{fragmentation} is $A=\ell+\rho$: exactly the number of nodes in the \cascade{}'s
interval map, since each live handle is one singleton node and each maximal
same-digest run is one node.
\end{definition}

\begin{theorem}[Fragmentation trade-off]
\label{thm:frag}
Write $A$ for the fragmentation (\Cref{def:frag}).
\begin{enumerate}[label=\textup{(\alph*)},leftmargin=2.2em,itemsep=1pt]
\item \emph{Upper bounds.} The \cascade{} stores $\Theta(A)$ index nodes; it resolves
  a point in $\bigoh(\log A)$ and reports a range of $k$ handles in
  $\bigoh(\log A + k)$ (expected; worst case with a deterministic balanced tree); and
  over $a$ appends and $s$ supersedes the total update work is
  $\bigoh\big((a/B + s)\log A\big)$, with $A=\bigoh(1+a/B+s)$ throughout.
\item \emph{Space lower bound.} Any structure that resolves every handle and reports
  ordered handle ranges must store $\Omega(A)$ words of index on a state of
  fragmentation $A\le n^{1-\varepsilon}$, even if it summarises payloads arbitrarily.
\item \emph{Query lower bound.} Reporting the $k$ handles of a range requires
  $\Omega(\log A + k)$ (the $\log A$ by predecessor search among the $A$ runs,
  Proposition~\ref{prop:order}).
\item \emph{Adversary.} For every block size $B$ and every $s\le n/2$, an oblivious
  sequence of $s$ supersedes against any structure with $o(n)$ materialised payloads
  drives the fragmentation to $A=\Theta(s)$; the maximum $A=\Theta(n)$ is reached at
  $s=\Theta(n)$.
\end{enumerate}
\end{theorem}
\begin{proof}
(a) The interval map has one node per live handle and one per same-digest run, so $A$
nodes and height $\bigoh(\log A)$; $\textsc{stab}$, $\textsc{append}$,
$\textsc{materialise}$, and $\textsc{summarise}$ each walk one root-to-leaf path, and
range reporting visits $\bigoh(\log A)$ boundary nodes plus the $\bigoh(k)$ reported
handles, expanding summarised runs in place. The update total is \Cref{rem:potential}
with $\log n$ read as the live height $\log A$: the $a$ appends fold in $\bigoh(a/B)$
splices and the $s$ supersedes in $\bigoh(s)$ splices, each $\bigoh(\log A)$. A
supersede raises $A$ by at most two and a fold lowers it, so $A=\bigoh(1+a/B+s)$.

(b) Consider states in which the $n$ handles are cut into $A$ maximal runs with
adjacent runs in distinct digests, the $A-1$ boundaries placed at any $A-1$ of the
$n-1$ gaps. There are $\binom{n-1}{A-1}\ge (n/A)^{A-1}$ placements, so the family has
$\Omega(A\log(n/A))$ states. A structure meeting the premises distinguishes them: a
range query over $[1,i]$ together with the resolutions of its endpoints reveals the
digest on each side of gap $i$, so scanning the gaps recovers every boundary and hence
the state; the payload summaries play no part. Its memory is therefore an injective
code for the family and occupies $\Omega(A\log(n/A))=\Omega(A)$ words of
$\Theta(\log n)$ bits for $A\le n^{1-\varepsilon}$. The $\Omega(A)$ index is met by
the \cascade{}'s $A$ nodes.

(c) To output $k$ handles costs $\Omega(k)$; locating the range endpoint among the
$A$ runs is predecessor search, not $\bigoh(1)$ in the cell-probe model within this
space (Proposition~\ref{prop:order}), giving the additive $\Omega(\log A)$.

(d) Append $n$ handles and fold them into one run ($A=1$). Now supersede the handles
at positions $\lfloor in/s\rfloor$, $i=0,\dots,s-1$, for $s\le n/2$ so consecutive
edits are at least two handles apart. Each supersede materialises a handle interior to
a summarised run, splitting it; with only $o(n)$ payloads materialisable the handle is
demoted again before the sweep returns to its neighbourhood and is re-folded into a
\emph{fresh} digest, so the split it made survives as a distinct-digest boundary.
Because the edits are spaced, the boundaries do not merge, leaving $\Theta(s)$ of them,
so $A=\Theta(s)$; the largest $A=\Theta(n)$ occurs at $s=\Theta(n)$, while editing
\emph{every} handle re-coalesces the run and lowers $A$ again---graceful, and visible
in \Cref{fig:frag}.
\end{proof}

\begin{corollary}[Instance optimality of the index]
\label{cor:instance}
On every operation sequence the \cascade{}'s index is within a constant factor of the
smallest any reference-stable, ordered, bounded-working-set structure can use: it is
$\Theta(A)$, and $\Omega(A)$ is forced (\Cref{thm:frag}b). The index is therefore
$o(n)$ exactly when the history is append-dominated ($A=\bigoh(a/B)$) and grows to
$\Theta(s)$ only when the edits are adversarially fragmenting---and then so must any
competitor's. The empirical index/$n$ curve of \Cref{fig:sweep} and the adversarial
\Cref{fig:frag} trace $A$ directly.
\end{corollary}

\begin{theorem}[External-memory I/O complexity]
\label{thm:io}
Realise the interval map as a persistent $(a,b)$-tree with $a=\lceil\beta/2\rceil$,
$b=\beta$, one page per node holding up to $\beta$ interval entries, and store each
digest as the $\lceil B'/\beta\rceil$ consecutive pages holding the
$(\text{version},\sigma(\text{payload}))$ records of its run of length $B'$. In the
I/O model with page size $\beta$~\citep{AggarwalVitter1988}:
\begin{enumerate}[label=\textup{(\roman*)},leftmargin=2em,itemsep=1pt]
\item $\resolve$ of a live handle costs $\bigoh(\log_\beta A)$ I/Os; of a summarised
  handle, $\bigoh(\log_\beta A)$ to stab its interval node plus one digest-page read.
\item reporting a $k$-handle range costs $\bigoh(\log_\beta A + k/\beta)$ I/Os: the
  $\bigoh(\log_\beta A)$ boundary pages, then for each summarised run it meets the
  $\lceil B'/\beta\rceil$ digest pages holding that run's handles, which sum to
  $\bigoh(k/\beta)$.
\item a supersede ($\textsc{materialise}$) and a fold ($\textsc{summarise}$) are
  $\bigoh(1)$ $(a,b)$-tree split/merge restructurings, each copying the
  $\bigoh(\log_\beta A)$ pages on one root-to-leaf path for persistence---so
  $\bigoh(\log_\beta A)$ page writes, \emph{independent of the run length $B'$}, where
  a per-handle scheme writes $\Theta\!\big(\tfrac{B'}{\beta}\log_\beta n\big)$ pages.
\item the $C+B$ live payloads occupy $\bigoh\big((C+B)c_{\max}/\beta\big)$ hot pages
  and the digests the cold pages; a retained snapshot freezes the $\bigoh(\log_\beta A)$
  pages on the copied path.
\end{enumerate}
\end{theorem}
\begin{proof}
Heights and split/merge page counts are the standard $(a,b)$-tree bounds for $A$
keys~\citep{BayerMcCreight1972,Comer1979btree}. A summarised handle's interval entry
records its digest's first page and the handle's in-run offset, so its record is read
in one further I/O; a range meeting a run reads that run's $\lceil B'/\beta\rceil$
digest pages, and the handles reported over all runs total $k$, giving $\bigoh(k/\beta)$
beyond the boundary path. Node copying rewrites exactly the touched path, and
coalescing makes a fold one restructuring per run. Buffering updates lowers writes to
the $\bigoh(\tfrac1\beta\log_\beta n)$ amortised regime of write-optimised
indexes~\citep{ONeil1996lsm,BrodalFagerberg2003} at the usual read cost.
\end{proof}

\begin{theorem}[Budgeted window approximation]
\label{thm:window}
Let $S^\star$ maximise $f(S)=\sum_{v\in\cover(S)}w(v)$ subject to
$\sum_{u\in S}c(u)\le\beta$. \Cref{alg:window} returns a feasible set with
$f(S)\ge\tfrac12(1-1/e)\,f(S^\star)$, and \Cref{alg:window-partial} returns a
feasible set with $f(S)\ge(1-1/e)\,f(S^\star)$. If $f$ is modular (records reveal
disjoint context), prefix-by-rank is optimal under a cardinality budget.
\end{theorem}
\begin{proof}
$f$ is non-negative, monotone---adding a record only enlarges $\cover(S)$---and
submodular, as the marginal $\sum_{v\in\cover(u)\setminus\cover(S)}w(v)$ is
non-increasing in $S$; with weights $w\ge 0$ it is a weighted coverage function,
so the instance is budgeted maximum coverage. For such instances the modified
greedy of \Cref{alg:window} attains $\tfrac12(1-1/e)$ and the size-three partial
enumeration of \Cref{alg:window-partial} attains
$1-1/e$~\citep{KhullerMossNaor1999}, the cardinality tightness of the latter being
classical~\citep{NemhauserWolseyFisher1978}. If $f$ is modular,
$f(S)=\sum_{u\in S}g(u)$ with $g(u)=w(u)$, and a cardinality budget is maximised by
the highest-value items---prefix-by-rank.
\end{proof}

\begin{proposition}[Prefix-by-rank has no constant-factor guarantee]
\label{prop:adversarial}
For every $k$ there is a budgeted maximum-coverage instance on which prefix-by-rank
covers a $\tfrac{1+\epsilon}{k}$ fraction of the optimum, while \Cref{alg:window}
and \Cref{alg:window-partial} are optimal.
\end{proposition}
\begin{proof}
Take one \emph{decoy} record of relevance $1+\epsilon$ and cost $\beta$ that reveals
only itself, and $k$ records of relevance $1$ and cost $\beta/k$ that reveal
pairwise-disjoint elements. Prefix-by-rank orders by relevance, takes the decoy
first, exhausts the budget, and covers $1+\epsilon$. The $k$ cheap records have
density $k/\beta$, above the decoy's $(1+\epsilon)/\beta$, so both greedies take
them and cover $k$, which is also the optimum; the ratio for prefix is
$(1+\epsilon)/k\to 0$. \Cref{sec:empirical} confirms the gap numerically.
\end{proof}

Because the recent records are the cheapest to resolve and the most often
queried, the \cascade{} also exhibits the working-set behaviour studied for
self-adjusting structures~\citep{SleatorTarjan1985splay,SleatorTarjan1985amortized}:
queries that stay within the hot stratum touch only materialised payloads.

\section{Implementation}
\label{sec:impl}

The reference implementation is a self-contained Python package. The \VII{} is a
path-copying treap (standard library only); the \cascade{} adds the strata,
folding, digests, and windows. Four comparison structures share the operation
interface. \textsf{Flat} keeps every version materialised in one table.
\textsf{NaiveTier} keeps a bounded hot table over a cold tier that is folded by
discarding all per-handle state; \textsf{ForwardTier} adds a never-reclaimed
per-handle version counter, which restores version authority but not
addressability. The progression \textsf{NaiveTier}$\to$\textsf{ForwardTier}$\to$\cascade{}
isolates the two mechanisms---version authority and post-fold addressability.
\textsf{HashSpine} is the strong baseline: it is the \cascade{} stripped of
persistence and order, a single hash map that is the sole version authority and
retains the same constant-size spine through folding. It is therefore anomaly-free
and bounded in working set, and resolves in expected $\bigoh(1)$, matching the
\cascade{} on every current-state metric; what it cannot do is read a past epoch
(it keeps no version history) or report an ordered key range across strata in
$\bigoh(\log n+k)$ (a hash imposes no order). It isolates precisely what the
\emph{persistent, ordered} index buys. An independent oracle tracks the exact
current and historical version of every handle and classifies each resolution by
Definition~\ref{def:anomaly}. All workloads and measurements are seeded and reproduce
bit-for-bit; the package, the result files, and the figure scripts accompany the
paper as ancillary files.

\section{Empirical study}
\label{sec:empirical}

\paragraph{Workload.}
A workload is an append-mostly stream: each step appends a record---with a few
references drawn toward recent handles with a heavy tail to old ones---or, with
probability $0.3$, supersedes an existing handle chosen with a bias toward older
records, so that edits reach handles that have already been folded. Record costs
are log-normal. Accesses follow a mixture that is local (recency-biased) with
probability $0.75$ and global (uniform over the whole history) otherwise,
modelling references both to the recent working set and to long-lived anchors.
Unless noted, the hot capacity is $C=256$ and the fold block $B=64$.

\paragraph{Reference integrity.}
\Cref{tab:headline} reports a run of $4\times10^4$ operations
($n=27{,}894$ handles) under $5\times10^4$ accesses. The tiered baselines
\textsf{NaiveTier} and \textsf{ForwardTier} keep a small footprint yet mis-resolve
about a quarter of accesses and leave only $1.3\%$ of the history addressable;
\textsf{ForwardTier}'s version counter removes the stale reads that afflict
\textsf{NaiveTier} but cannot recover the folded records. \textsf{Flat} is correct
but materialises every payload. \textsf{HashSpine} and the \cascade{} run the
identical ageing and folding policy: both are anomaly-free, fully addressable, and
bounded in working set, with the same $7{,}038$-unit working set against
\textsf{Flat}'s $6.7\times10^5$, so reference integrity costs essentially nothing.
They differ in the \emph{index}. The hash spine keeps one entry per handle
($27{,}894$ nodes); the \cascade{} coalesces folded runs and keeps $10{,}983$, so
its total footprint ($45{,}608$) is below the hash spine's ($62{,}519$) even though
it additionally serves snapshot-consistent past-epoch reads and ordered cross-tier
ranges (\Cref{tab:capabilities}). The ordered, persistent index is not a cost paid
for those capabilities; here it is cheaper than the hash that lacks them.

\begin{table}[t]
\centering
\caption{One operating point ($n=27{,}894$ handles, $5\times10^4$ accesses,
$C{=}256$, $B{=}64$). ``Addressable'' is the fraction of the whole history that
still resolves; ``working set'' is the materialised payload; ``index'' is the
number of index nodes; ``total'' is the whole materialised footprint (working set,
index, and cold digests), in abstract units.}
\label{tab:headline}
\small
\begin{tabular}{l r r r r r}
\toprule
Design & Anomaly rate & Addressable & Working set & Index nodes & Total \\
\midrule
\textsf{Flat}        & $0.000$ & $1.000$ & $670{,}507$ & $27{,}894$ & $670{,}507$ \\
\textsf{NaiveTier}   & $0.249$ & $0.013$ & $13{,}295$  & $371$ & $13{,}295$ \\
\textsf{ForwardTier} & $0.248$ & $0.013$ & $41{,}189$  & $371$ & $41{,}189$ \\
\textsf{HashSpine}   & $0.000$ & $1.000$ & $7{,}038$ & $27{,}894$ & $62{,}519$ \\
\textbf{\cascade{}}  & $\mathbf{0.000}$ & $\mathbf{1.000}$ & $\mathbf{7{,}038}$ & $\mathbf{10{,}983}$ & $\mathbf{45{,}608}$ \\
\bottomrule
\end{tabular}
\end{table}

\begin{table}[t]
\centering
\caption{Capabilities at a bounded working set. Among designs that bound the
working set, only the \cascade{} is anomaly-free \emph{and} serves
snapshot-consistent past-epoch reads \emph{and} reports ordered cross-tier ranges;
the price is $\bigoh(\log n)$ rather than expected $\bigoh(1)$ resolution.}
\label{tab:capabilities}
\small
\begin{tabular}{l c c c c c}
\toprule
Design & Anomaly-free & Bounded W.S. & Snapshots & Ordered range & Resolve \\
\midrule
\textsf{Flat}        & yes & no  & no  & no  & $\bigoh(1)$ \\
\textsf{NaiveTier}   & no  & yes & no  & no  & $\bigoh(1)$ \\
\textsf{ForwardTier} & no  & yes & no  & no  & $\bigoh(1)$ \\
\textsf{HashSpine}   & yes & yes & no  & no  & $\bigoh(1)$ \\
\textbf{\cascade{}}  & \textbf{yes} & \textbf{yes} & \textbf{yes} & \textbf{yes} & $\bigoh(\log n)$ \\
\bottomrule
\end{tabular}
\end{table}

\Cref{fig:anomaly} sweeps the hot capacity. As the working set tightens, both
tiered baselines lose addressability---from $59\%$ down to under $0.5\%$ of the
history---and their access-anomaly rate climbs toward the share of accesses that
reach beyond the hot stratum, while the \cascade{} stays at zero anomalies and
full addressability throughout. The separation between \textsf{NaiveTier} and
\textsf{ForwardTier} is exactly the stale component, which the version counter
removes.

\begin{figure}[t]
\centering
\includegraphics[width=\linewidth]{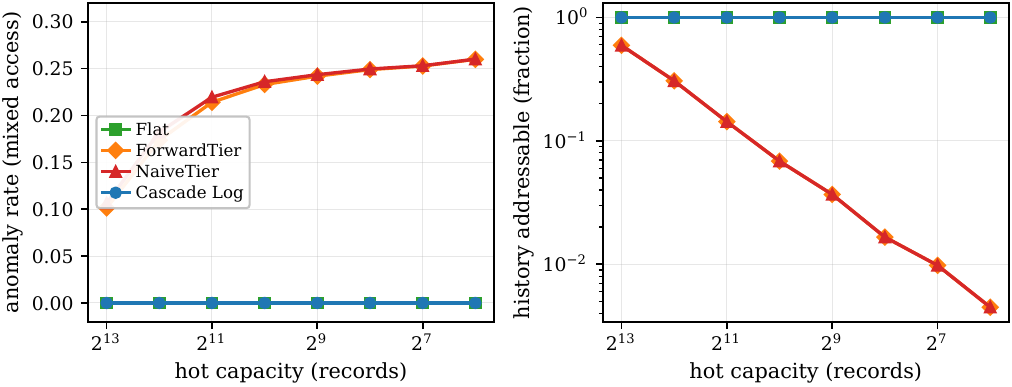}
\caption{Reference integrity as the hot capacity tightens (capacity decreases to
the right). Left: anomaly rate under the mixed access stream. Right: fraction of
the whole history that still resolves (log scale). \textsf{Flat} and the
\cascade{} coincide at zero anomalies and full addressability; the tiered
baselines degrade.}
\label{fig:anomaly}
\end{figure}

\paragraph{Footprint and scaling.}
\Cref{fig:memory} grows the history to $n\approx4\times10^5$ (half a million
operations). \textsf{Flat} grows linearly at the average payload size. The
\cascade{}'s \emph{working set}---its materialised payload---stays flat at about
seven thousand units, independent of $n$, as \Cref{thm:complexity} predicts in
ephemeral mode; its total grows linearly but about nineteen times below full
materialisation at the largest size (a ratio that widens with record size), and
below \textsf{HashSpine} throughout because its index coalesces (next paragraph).
Retaining $R$ snapshot tokens raises the index from one node per handle to about
$3.8$ at $R=100$ (\Cref{fig:snapshot}), the retained-mode overhead of
\Cref{thm:complexity}, growing sub-linearly per token as tokens share structure.

\begin{figure}[t]
\centering
\begin{minipage}{0.6\linewidth}\centering
\includegraphics[width=\linewidth]{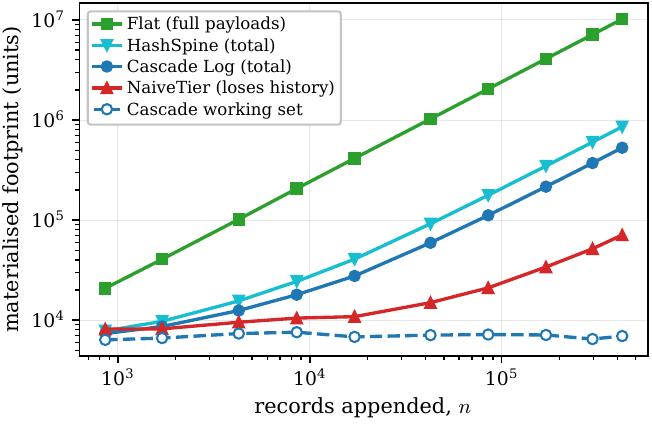}
\end{minipage}\hfill
\begin{minipage}{0.38\linewidth}\centering
\includegraphics[width=\linewidth]{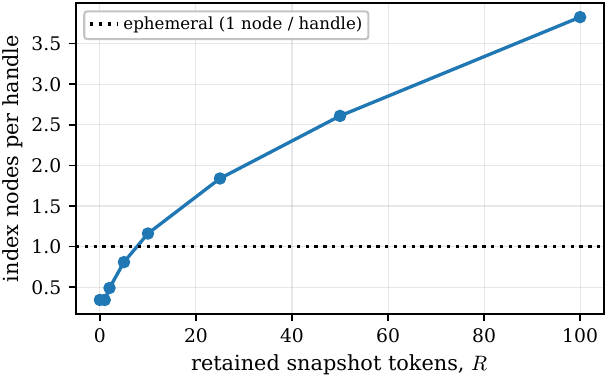}
\end{minipage}
\caption{Left: materialised footprint versus history length (log--log); the
\cascade{} working set (dashed) is bounded, and its total stays below
\textsf{HashSpine} and far below \textsf{Flat}. Right: retained-mode index size
against the number of retained snapshot tokens (\Cref{thm:complexity}).}
\label{fig:memory}
\label{fig:snapshot}
\end{figure}

\paragraph{Compaction cost and the index.}
\Cref{fig:fold} measures the index work of folding. With no edits a fold block is
a single contiguous run and collapses in exactly one splice, so the splices per
record fall as $1/B$ and the amortised index cost of compaction vanishes with the
block size---the $\bigoh(1)$-per-appended-record claim of
\Cref{thm:complexity}(b). Edits fragment runs, each fragmentation adding one
splice, so the curves rise with the edit rate but stay essentially flat in $B$.
The same coalescing shrinks the index (\Cref{fig:index}): the \cascade{} indexes
the history in a number of nodes that is a \emph{decreasing} fraction of $n$
(about $0.23$ here under a $15\%$ edit rate, and $1/B$ for append-only) against the
hash spine's one node per handle. The \cascade{} therefore folds in amortised
constant index work per appended record and indexes the history in $o(n)$ nodes,
where a per-handle forwarding design does $\Theta(B)$ work per block and keeps $n$
nodes.

\begin{figure}[t]
\centering
\begin{minipage}{0.49\linewidth}\centering
\includegraphics[width=\linewidth]{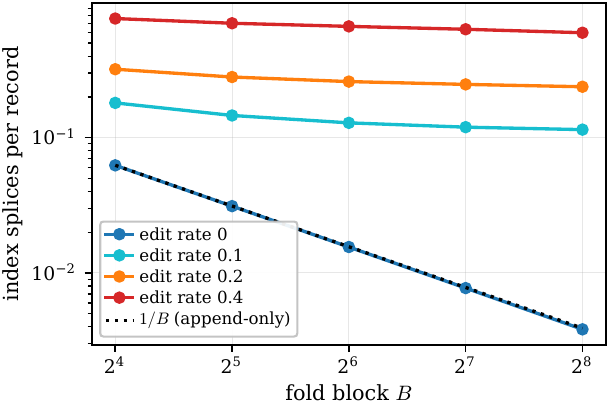}
\end{minipage}\hfill
\begin{minipage}{0.49\linewidth}\centering
\includegraphics[width=\linewidth]{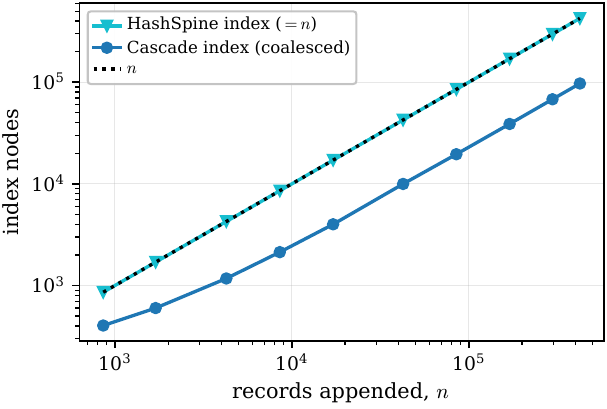}
\end{minipage}
\caption{Left: index splices per folded record versus block size $B$. Append-only
folding (edit rate $0$) lies on $1/B$---one splice per block---so its per-record
index work vanishes with $B$; edits add a flat, rate-proportional fragmentation
charge. Right: index nodes versus $n$; the \cascade{} coalesces folded runs to an
$o(n)$ index, while the hash spine keeps $n$ entries.}
\label{fig:fold}
\label{fig:index}
\end{figure}

\paragraph{A persistent ordered map, and ordered range.}
To isolate what \emph{coalescing} contributes---as opposed to ordering and
persistence, which a copy-on-write ordered map already provides---we add a fifth
structure, \textsf{PersistentMap}: the \cascade{}'s policy over a persistent treap
that keeps one node per handle and rewrites each folded handle individually, with
no coalescing; it is the in-memory analogue of a multiversion or copy-on-write
$B$-tree~\citep{BeckerMVBT1996}. It is anomaly-free and serves snapshots and ordered
range like the \cascade{}, but its index is $n$ nodes and each fold does
$\Theta(B\log n)$ index writes. \Cref{fig:pmap}~(left) confirms the gap: at $n=10^5$
the \cascade{}'s coalesced index is $4.3\times$ smaller and its folds do $6\times$
fewer index writes, with ordering and persistence held fixed---so coalescing is the
sole cause of the smaller index. \Cref{fig:pmap}~(right) measures the ordered
capability \textsf{HashSpine} lacks (Proposition~\ref{prop:order}): reporting a range
of $k$ handles, the \cascade{} costs $\bigoh(\log n + k)$ and tracks $k$, while
\textsf{HashSpine} must scan its whole table at $\approx 2.9$\,ms per query
independent of $k$---two to three orders of magnitude slower until $k$ nears $n$.

\begin{figure}[t]
\centering
\begin{minipage}{0.49\linewidth}\centering
\includegraphics[width=\linewidth]{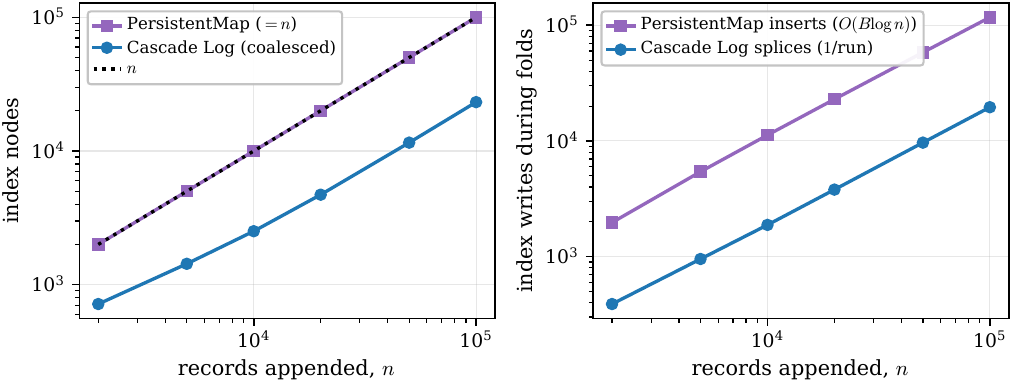}
\end{minipage}\hfill
\begin{minipage}{0.49\linewidth}\centering
\includegraphics[width=\linewidth]{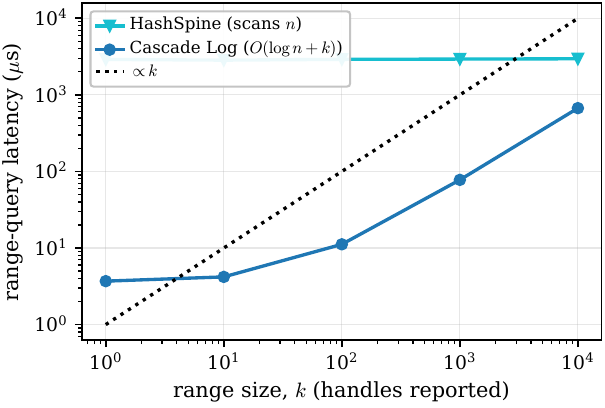}
\end{minipage}
\caption{Left: index nodes versus $n$ for the \cascade{} and \textsf{PersistentMap}
(an ordered, persistent map without coalescing); coalescing alone accounts for the
gap. Right: ordered-range latency versus the number of handles reported $k$; the
\cascade{} scales as $k$ while the order-less \textsf{HashSpine} pays a full scan.}
\label{fig:pmap}
\label{fig:range}
\end{figure}

\paragraph{Robustness across the edit rate.}
\Cref{fig:sweep} repeats the integrity and index measurements across supersede
probabilities from $0.1$ to $0.5$, five seeds each, with $95\%$ confidence
intervals. The \cascade{} and \textsf{HashSpine} hold exactly zero anomalies
throughout, while \textsf{NaiveTier} and \textsf{ForwardTier} sit near $0.25$ with
tight spread. The right panel is the honest face of the index claim: the
\cascade{}'s index/$n$ ratio climbs from $0.18$ at a $10\%$ edit rate to $0.58$ at
$50\%$, well above the append-only floor $1/B$. The index is $o(n)$ only when
appends dominate, and a smaller-constant $\Theta(n)$ under sustained edits---exactly
the qualified claim, shown rather than asserted.

\begin{figure}[t]
\centering
\includegraphics[width=\linewidth]{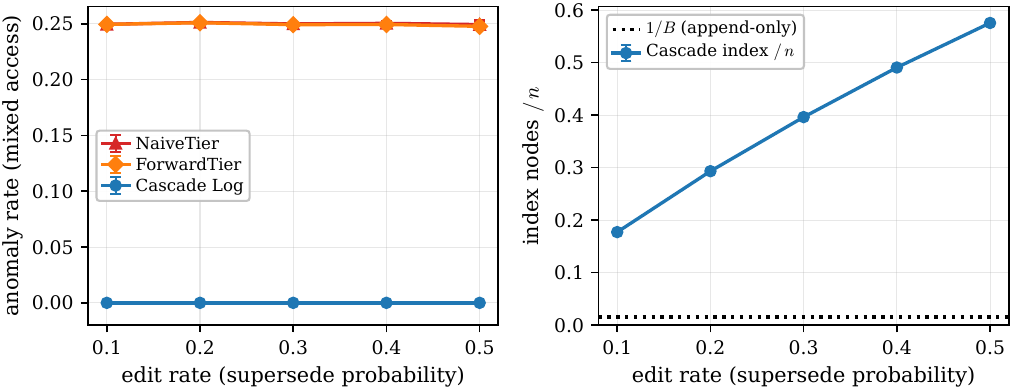}
\caption{Robustness over five seeds per point ($95\%$ confidence intervals). Left:
anomaly rate versus edit rate; the \cascade{} stays at zero. Right: the \cascade{}
index as a fraction of $n$ rises with the edit rate (append-only floor $1/B$
dotted), so the $o(n)$ index is an append-dominance property, not a universal one.}
\label{fig:sweep}
\end{figure}

\paragraph{Adversarial fragmentation.}
\Cref{fig:frag} drives the worst case of \Cref{thm:frag}: after appending
$6\times10^4$ records it supersedes $s$ handles at evenly spaced positions---the
pattern that maximises fragmentation---and traces the index size directly. It grows
\emph{linearly} in $s$, $A=\Theta(s)$ at about two nodes per edit, for both the spread
adversary and a distribution-free random-position variant, rising toward the $n$-node
ceiling near $s\approx n/2$ and never exceeding \textsf{PersistentMap}'s $n$; once
nearly every handle has been edited the re-materialised handles re-fold contiguously
and $A$ collapses back toward $\bigoh(a/B)$. The right panel confirms resolution stays
$\bigoh(\log A)$
even at maximum fragmentation, tracking $\log_2 A$. The structure degrades gracefully
and exactly to the bound: its index \emph{is} the fragmentation---no more, and by
\Cref{thm:frag}(b) no less than any reference-stable ordered structure must pay.

\begin{figure}[t]
\centering
\includegraphics[width=\linewidth]{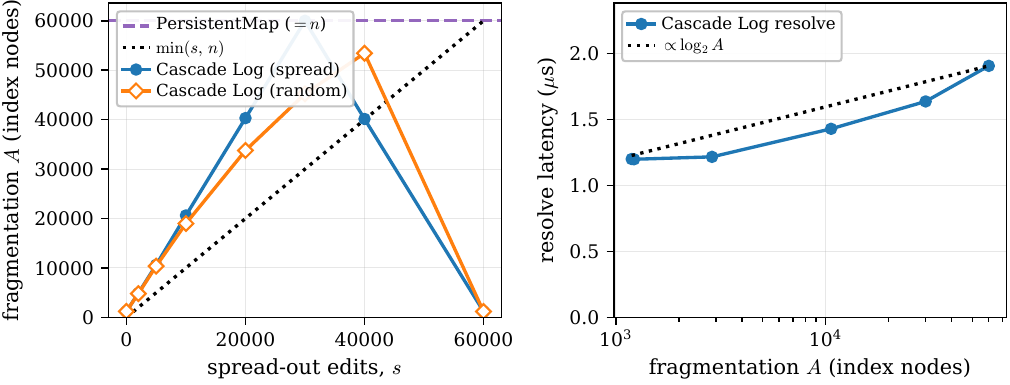}
\caption{Adversarial fragmentation (\Cref{thm:frag}). Left: index nodes $A$ versus the
number of evenly spaced (and randomly placed) edits $s$; the \cascade{} grows as
$\Theta(s)$ and never exceeds \textsf{PersistentMap}'s $n$. Right: resolution latency
versus $A$, tracking $\log_2 A$, so queries stay logarithmic even at worst-case
fragmentation.}
\label{fig:frag}
\end{figure}

\paragraph{Latency.}
\Cref{fig:latency} (left) reports resolution latency to $n=10^6$. The \cascade{}
resolves in $0.7$--$4.1\,\mu$s, scaling as the logarithm of its coalesced index;
the index height stays below the $2\log_2 n$ guide (right panel) precisely because
folded runs collapse. \textsf{HashSpine} resolves by one hash probe in about
$0.8\,\mu$s at $n=10^6$---a constant factor below the \cascade{}, the price it pays
for the ordered, snapshot-consistent access a hash cannot give. Append is a single
logarithmic splice per operation.

\begin{figure}[t]
\centering
\includegraphics[width=\linewidth]{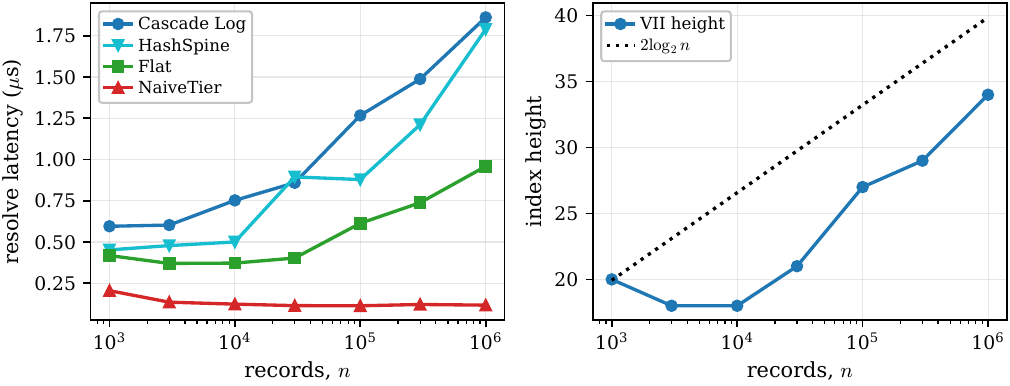}
\caption{Left: resolution latency versus history length to $n=10^6$ (log $x$);
\textsf{HashSpine} is a constant factor below the \cascade{}, the tiered baselines
faster still but unable to read the past or, for \textsf{NaiveTier}, correctly.
Right: \VII{} height against a $2\log_2 n$ guide---logarithmic with a small
constant.}
\label{fig:latency}
\end{figure}

\paragraph{Windows.}
\Cref{fig:window} evaluates window selection on a reference graph of $4000$
records, averaging over recent anchors and budgets expressed as a fraction of the
candidate cost. The partial-enumeration of \Cref{alg:window-partial} is
indistinguishable from the exact optimum (computed by brute force and a
mixed-integer program), staying above $0.99$ of it and---like the simple greedy of
\Cref{alg:window}, whose worst case here is $0.85$---never below its guarantee. The
prefix heuristic trails by up to thirteen points of coverage at tight budgets. The
gap is unbounded in the worst case: on the constructed family of
Proposition~\ref{prop:adversarial}, \Cref{fig:adversarial} shows prefix-by-rank falling as
$(1+\epsilon)/k$ while both greedies stay at the optimum. Simple-greedy windows are
selected in under a millisecond; the partial enumeration costs tens of milliseconds
on these candidate sets.

\begin{figure}[t]
\centering
\begin{minipage}{0.66\linewidth}\centering
\includegraphics[width=\linewidth]{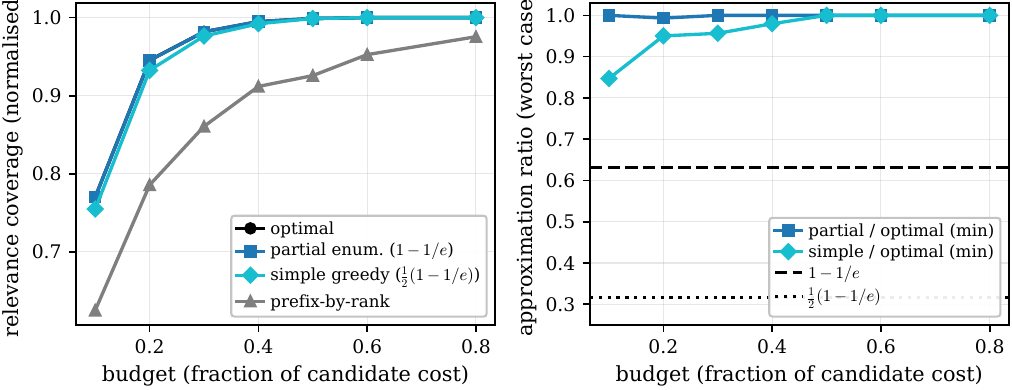}
\end{minipage}\hfill
\begin{minipage}{0.32\linewidth}\centering
\includegraphics[width=\linewidth]{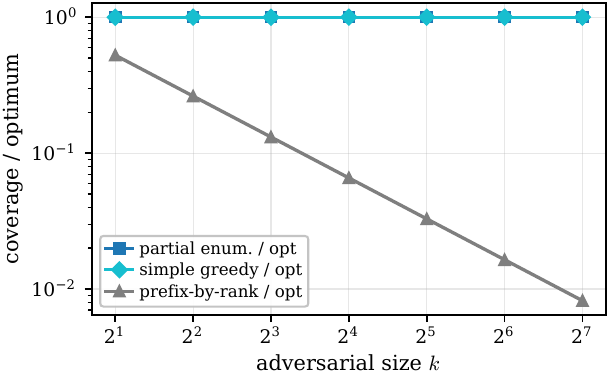}
\end{minipage}
\caption{Budgeted windows. Left pair: relevance coverage and worst-case
approximation ratio of prefix-by-rank, the simple greedy ($\tfrac12(1-1/e)$), the
partial enumeration ($1-1/e$), and the optimum. Right: on the family of
Proposition~\ref{prop:adversarial}, prefix-by-rank is a $\Theta(1/k)$ fraction of the optimum
while both greedies are optimal (log--log).}
\label{fig:window}
\label{fig:adversarial}
\end{figure}

\paragraph{Summary.}
Among designs that bound the working set, \textsf{HashSpine} and the \cascade{} are
the anomaly-free ones; the tiered baselines that are not lose almost the entire
history. Against the hash spine the \cascade{} gives up a constant factor in
resolution latency and, in return, coalesces its index---to $o(n)$ nodes on
append-dominated streams, and a smaller constant times $n$ under sustained edits,
in both regimes below the hash spine's $n$---folds in amortised constant index work
per appended record, and serves snapshot-consistent past-epoch reads and ordered
cross-tier ranges that a hash cannot. \Cref{thm:lower} shows the per-handle version
both retain cannot be lossily compressed, and Proposition~\ref{prop:order} that the
ordered range the \cascade{} adds is provably beyond a hash.

\section{Related work}
\label{sec:related}

\paragraph{Persistence.}
Path copying and node copying make pointer structures persistent at logarithmic
or constant amortised overhead~\citep{DriscollSarnakSleatorTarjan1989}, and
persistent search trees underlie classical multiversion
queries~\citep{SarnakTarjan1986}. The \VII{} uses persistence not for geometric
queries but to make each fold and edit produce an immutable instant that windows
can be read against; the snapshot consistency of \Cref{thm:linear} follows from
persistence rather than from locking, in the linearizable
sense of~\citet{HerlihyWing1990}.

\paragraph{Tiered and write-optimised storage.}
Folding aged blocks into summaries is the organising idea of the log-structured
merge-tree~\citep{ONeil1996lsm}, where membership filters~\citep{Bloom1970} avoid
fruitless probes into cold levels and amortised analysis bounds the merge
cost~\citep{Tarjan1985amortized,BentleySaxe1980}. These designs optimise
throughput for opaque keys; the \cascade{} adds what they do not track---a stable
per-record handle that must keep resolving across the very compaction that moves
its target.

\paragraph{External memory.}
The motivating setting is tiered storage, and \Cref{thm:io} carries the design into
the I/O model~\citep{AggarwalVitter1988}: an $(a,b)$-tree
realisation~\citep{BayerMcCreight1972,Comer1979btree} resolves and folds in
$\bigoh(\log_\beta n)$ I/Os and reports a range in $\bigoh(\log_\beta n + k/\beta)$,
bounded below by the dictionary update--query trade-off
of~\citet{BrodalFagerberg2003} that the log-structured merge-tree also navigates.
Our measurements are RAM-resident (\Cref{sec:empirical}); a paginated implementation
and a head-to-head against tuned LSM and B-tree stores are the systems counterpart
this paper does not attempt.

\paragraph{Multiversion and temporal indexing.}
Serving a read as of a past instant is the business of multiversion concurrency
control~\citep{BernsteinGoodman1983} and of versioned and temporal indexes---the
asymptotically optimal multiversion B-tree~\citep{BeckerMVBT1996} and the access
methods surveyed by~\citet{SalzbergTsotras1999}---which answer time-travel queries
at scale and concurrently. These index opaque keys and keep every version; the
\cascade{} differs in the two ways that define its niche. It bounds the materialised
working set by \emph{lossily summarising} aged records (\Cref{def:anomaly}), trading
payload fidelity for a footprint independent of history length, where a
multiversion B-tree retains every version in full; and it guarantees
\emph{reference integrity}, a stable handle resolving to the right version across the
compaction that summarises its target, which the versioned-index literature does not
pose because its keys neither migrate nor summarise. The two are complementary: a
multiversion B-tree is one concurrent, external-memory realisation of the persistent
ordered map the \cascade{} needs (\Cref{thm:io}), and the \cascade{}'s own
contribution is the coalescing fold and the cross-tier guarantees layered on it.

\paragraph{Layered search.}
Searching ordered layers in one logarithmic pass is the province of fractional
cascading~\citep{ChazelleGuibas1986} and of probabilistic layered
lists~\citep{Pugh1990skiplists}; B-trees give the external-memory
counterpart~\citep{Comer1979btree}. The \cascade{}'s cross-tier reporting reaches
the same $\bigoh(\log n + k)$ through a single multiversion index, and any of
these techniques may realise it over per-stratum indexes instead.

\paragraph{Ranking and selection.}
Diffusing relevance by a personalised random walk is personalised
PageRank~\citep{BrinPage1998,Haveliwala2003}; choosing a budgeted subset that
maximises a submodular coverage is budgeted maximum
coverage~\citep{KhullerMossNaor1999}, whose greedy inherits the classical
$1-1/e$ analysis for submodular maximisation~\citep{NemhauserWolseyFisher1978}.
The window of \Cref{sec:window} composes the two over the reference graph and
recovers prefix-by-rank as the modular special case.

\paragraph{Sequence buffers and diffs.}
Editor buffers that keep a long mutable sequence in balanced
pieces~\citep{BoehmAtkinsonPlass1995ropes} and difference algorithms that compute
edits between versions~\citep{Myers1986diff} address the per-record content that a
\cascade{} record may hold; they are orthogonal to, and composable with, the
cross-tier addressing studied here.

\section{Conclusion}
\label{sec:conclusion}

A tiered append sequence keeps its working set bounded, but the obvious way to do
so discards the per-record state that lets handles resolve across tier
boundaries, and the result fails exactly where the tiers meet. The \cascade{}
shows that the two goals are compatible at a price that is provably small: a
persistent interval map as the single authority on every handle, folding that
coalesces a run into one node in a single splice---so compaction is amortised
constant per appended record and the index is sublinear---and digests that keep
each handle's version, the information a lower bound shows is irreducible. The
result is resolution that is total, snapshot-consistent, and anomaly-free at
logarithmic cost, with a working set independent of history length; the budgeted
window is a standard maximum-coverage instance layered on top. The reference
implementation is single-threaded and pays the constant factor of a pure-language
persistent tree; a concurrent, de-amortised, external-memory realisation
(\Cref{sec:related}), and digests that compress versions succinctly, are natural
next steps.

\paragraph{Reproducibility.}
The implementation, the seeded workloads, the oracle, the benchmark harness, the
raw result files, and the figure scripts are provided as ancillary files; a single
command regenerates every number, table, and figure (\Cref{app:repro}).

\appendix
\section{Reproduction details}
\label{app:repro}
The ancillary directory \texttt{anc/} contains a self-contained Python package
(\texttt{numpy}, \texttt{scipy}, \texttt{matplotlib}, \texttt{pytest}; the core
structure uses only the standard library). From \texttt{anc/}, the command
\texttt{./run\_all.sh} runs the test suite (\texttt{python -m pytest cascade/tests},
$86$ checks), the benchmark harness (\texttt{python -m cascade.bench}), the extended
study (\texttt{python -m cascade.extra}), and the figure scripts
(\texttt{python -m cascade.plots}). Every workload is seeded, so the numbers
reproduce bit-for-bit; the seeds are the literal arguments to \texttt{gen\_plan} in
\texttt{cascade/bench.py} and \texttt{cascade/extra.py}. Each experiment writes one
CSV to \texttt{anc/results/} with a header row naming every column, and each figure
reads exactly one such file; \Cref{tab:repro} maps result files to artifacts.
Runtimes were measured on an Apple silicon laptop with CPython~$3.x$; the harness
completes in roughly thirteen minutes and the extended study in about half a minute.
The implementation is a single-threaded reference in a managed language, so absolute
latencies carry the constant factor of CPython and are reported only relative to the
in-house baselines run identically; the asymptotic claims, not the constants, are the
contribution.

\begin{table}[h]
\centering\small
\caption{Result files and the artifacts they produce.}
\label{tab:repro}
\begin{tabular}{l l}
\toprule
File (\texttt{anc/results/}) & Artifact \\
\midrule
\texttt{headline.csv}            & \Cref{tab:headline} \\
\texttt{anomaly\_vs\_capacity.csv} & \Cref{fig:anomaly} \\
\texttt{scaling.csv}             & \Cref{fig:memory}, \Cref{fig:index} \\
\texttt{fold.csv}                & \Cref{fig:fold} \\
\texttt{latency.csv}             & \Cref{fig:latency} \\
\texttt{window.csv}, \texttt{adversarial.csv} & \Cref{fig:window} \\
\texttt{snapshot\_overhead.csv}  & \Cref{fig:snapshot} \\
\texttt{pmap.csv}, \texttt{range.csv} & \Cref{fig:pmap} \\
\texttt{sweep.csv}               & \Cref{fig:sweep} \\
\texttt{frag\_vs\_edits.csv}, \texttt{frag\_latency.csv} & \Cref{fig:frag} \\
\bottomrule
\end{tabular}
\end{table}

{\small
\bibliographystyle{plainnat}
\bibliography{refs}
}

\end{document}